\documentclass[manuscript,nonacm]{acmart}

\usepackage{mathtools}

\title[Symmetric Bounded Indistinguishability]{Symmetric Bounded Indistinguishability: Hypergeometric Smoothing and Hahn Polynomials}

\author{Christopher Williamson}
\affiliation{
  \institution{Independent researcher}
  \city{Hong Kong}
  \country{Hong Kong, China}
}
\email{cw@williamsonchris.com}

\begin{abstract}
A pair of probability distributions over $\{0,1\}^n$ is said to be $(k,\delta)$-wise indistinguishable if all of the size $k$ marginals are within statistical distance at most $\delta$. Previous works introduce this concept and study how far apart $t$-wise marginals can be under an assumption of $(k,\delta)$-wise indistinguishability. We consider symmetric distributions and obtain a new upper bound that unifies and improves previous bounds and applies across a wider range of parameters. In particular, prior works failed to rule out the existence of constants $0<c<c'<1$ so that there is a pair of $(cn,0)$-wise indistinguishable distributions where the $c'n$-wise marginals have statistical distance $\Omega(1)$. Our upper bound shows that the $c'n$-wise marginals must be exponentially close for all $c$ and $c'$. Our upper bound is accompanied with a nearly matching lower bound when $\delta=0$ and also yields new results in the case $\delta>0$ or when $t/n$ tends to 1. Our approach is to exploit the behaviour of the orthogonal Hahn polynomials under hypergeometric sampling and marginalisation operations. As a secondary contribution, we provide nearly matching upper and lower bounds on the maximum possible distance between a pair of $(k,\delta)$-wise indistinguishable distributions and the nearest pair of $(k,0)$-wise indistinguishable distributions.
\end{abstract}



\acmJournal{TOCT}
\begin{document}
\maketitle

\section{Introduction}
We consider pairs of distributions ($\mu,\nu$) over $\{0,1\}^n$.
The distributions $\mu$ and $\nu$ are said to be $(k,\delta)$-wise
indistinguishable if for any subset $S \subseteq [n]$ of size at most $k$,
the marginal distributions $\mu_S$ and $\nu_S$ over indices in $S$ are
of statistical distance at most $\delta$. The distributions have $t$-wise statistical distance $\varepsilon$ if there exists a set $S \subseteq [n]$ of indices of size $t$ and a
statistical test $T:\{0,1\}^{t} \to \{0,1\}$ such that
\[
\left|
\mathbb E_{X \sim \mu}[T(X_S)]
-
\mathbb E_{Y \sim \nu}[T(Y_S)]
\right|
\geq \varepsilon,
\]
where $X_S$ is the restriction of random variable $X$ to the bits located
at the indices in $S$. The distributions are symmetric if $\mu$ and
$\nu$ are invariant under permutation (see Section~\ref{s:symmetry}); for
such distributions the size of $S$ is relevant but not
the choice of indices.

\paragraph{Cryptographic motivation.}

Work of Bogdanov et al.~\cite{bogdanov2016bounded, bogdanov2017approximate} considered this notion of
indistinguishability as a way to capture cryptographic secret sharing
schemes in a minimal setting. Their observation was that a single bit secret
can be shared by sampling $n$ bits from $\mu$ or from $\nu$, depending on
the secret: $(k,0)$-wise indistinguishability of the distributions provides
a perfect information theoretic security guarantee that any size $\leq k$ coalition of colluding parties
learns nothing about the secret from their joint shares. The secret
reconstruction function for the scheme is a test $T$ applied over the shares
of sufficiently many (possibly all) parties. Their work proves the existence of distributions over $\Sigma^n$ where $\Sigma$ consists of bit strings of length polynomial in $n$. In this work we only consider distributions over $\{0,1\}^n$.

The special case of $T=\mathrm{OR}$ captures the notion of visual
cryptography, introduced by Naor and Shamir~\cite{naor1994visual} in 1994. Their work
considered $(k,0)$-wise indistinguishable distributions and attempted to
maximise the reconstruction advantage of $\mathrm{OR}$ taken over $k+1$
bits. They gave a tight bound for all $n$ with $k=n-1$ and also derived a non-tight result for general $k$ and $n$. Krause and
Simon~\cite{krause2003determining} gave improved upper and lower bounds for general $k$ and $n$, but these bounds are exponentially far apart when $k=\Omega(n)$. In
this work, we consider the statistical distance between the distributions, which must capture the power of any given statistical test. 

\paragraph{Prior work.}
Huang and Viola~\cite{huang2022approximate} considered
the extent to which symmetric Boolean functions are able to distinguish between $(k,\delta)$-wise indistinguishable distributions. Because symmetric functions form a basis for any optimal statistical test over symmetric distributions, such results have the potential to be translated into an upper bound on $t$-wise statistical distance. However, their result only considers tests that observe all $n$ bits and does not provide a non-trivial upper bound on the distinguishing power of certain symmetric tests; thus they do not provide upper bounds on the statistical distance of marginal distributions. Huang and Viola also prove that if a pair of (not necessarily symmetric) distributions are \((k,\delta)\)-indistinguishable, then they are
\(O\!\left(e^k n^{k/2}\delta\right)\)-close to a pair of distributions that are \(k\)-wise indistinguishable.

Bogdanov et al.~\cite{bogdanov2019approximate} considered
the extent to which symmetric $(k,0)$-wise indistinguishable distributions must
have statistically close $t$-wise marginals for $t>k$. In particular, they show that if $\mu$ and $\nu$ are symmetric over $n$-bit
strings and $(k,0)$-wise indistinguishable, then the statistical
distance between $t$-wise marginals is at most
\begin{equation}
\label{eq:BMTW19}
    O(k^{3/2}) \cdot e^{-k^2/1156t}
\end{equation}
for all $k<t\leq n/64$.

In the sharp reconstruction setting where $t=k+1$, the above bound gives
\[
2^{-\Omega(t)}.
\]
This cannot possibly be tight for some $t>n/64$, for the simple reason that the parity function over all $n$ bits can distinguish perfectly between a pair of distributions that are $(n-1)$-wise indistinguishable (take the two distributions to be uniform but conditioned on the global parity being equal to 0 or 1). This observation motivated the work~\cite{williamson2021sharp}, which considers the sharp threshold case in which $t=k+1$. If for $x\in[0,1]$ we define the function 
\[
I(x):=\frac12\left((1+x)\log_2(1+x)+(1-x)\log_2(1-x)\right),
\]
where we take $I(1):=1$, then~\cite{williamson2021sharp} shows that if a pair of distributions are $(k,0)$-wise indistinguishable, then the $(k+1)$-wise marginals must be within statistical distance:
\begin{equation}
\label{eq:Wil22}
    n^{O(1)}\cdot 2^{-k+n\cdot I(k/n)},
\end{equation}
and this is sharp up to the leading polynomial factor.

These results leave basic questions unresolved; in particular, it is left open whether there exist constants $0<c<c'<1$ and a pair of symmetric distributions $\mu,\nu$ over $\{0,1\}^n$ that are $(cn,0)$-wise indistinguishable yet $c'n$-wise reconstructible with advantage $e^{-o(n)}$ or even $\Omega(1)$. 

\paragraph{Our contribution.}
We derive a new upper bound that generalises the results of~\cite{bogdanov2019approximate,williamson2021sharp}. Relative to~\cite{bogdanov2019approximate}, the new bound is sharper and allows $k$ and $t$ to be larger than $n/64$. Relative to~\cite{williamson2021sharp}, our new bound does not require that $t=k+1$. The new bound also provides results in the case $\delta>0$.  
\begin{theorem}
\label{thm:main-exact-H-form}
Let \(0\le s\le k<t< n\), and let \(\mu,\nu\) be symmetric
\((k,\delta)\)-wise indistinguishable distributions on \(\{0,1\}^n\).
Then, the statistical distance between all $t$-wise marginals is at most:
\[
    n^{O(1)}\cdot \frac{1}{2^{t\cdot I(s/t)}}\left( 2^{n\cdot I(s/n)}+\delta \cdot 2^{k\cdot I(s/k)}\right),
\]
\end{theorem}
The parameter $s$ can be freely chosen to minimise the upper bound; its optimal setting depends on $\delta$. When $\delta=0$, we set $s=k$ and arrive at a corollary:
\begin{corollary}
In the setting of Theorem~\ref{thm:main-exact-H-form}, with $\delta=0$, we have that the statistical distance between the $t$-wise marginals is at most
\[
    n^{O(1)}\cdot 2^{-t\cdot I(k/t)+n\cdot I(k/n)},
\]
\end{corollary}
which is easy to compare against the~\cite{williamson2021sharp} bound (\ref{eq:Wil22}). This corollary comes with a nearly matching lower bound, stated and proved in Section~\ref{s:mainLBsection}.

A useful restatement of Theorem~\ref{thm:main-exact-H-form}, justified in Section~\ref{s:subsubsectionCor} is:
\begin{corollary}
\label{cor:main-cor-delta}
In the setting of Theorem~\ref{thm:main-exact-H-form}, the statistical distance between the $t$-wise marginals is at most
\begin{equation}
\label{eq:working-approx-bound}
    \exp\!\left(
        \frac12\log(n+1)
        -
        \frac{(n-t)s^2}{6nt}
    \right)
    +
    \delta\,\exp\!\left(\frac12\log(k+1)+
        \frac{(t-k)s(s+1)}
             {2t(k-s+1)}
    \right).
\end{equation}
\end{corollary}

Corollary~\ref{cor:main-cor-delta} is useful for recovering bound (\ref{eq:BMTW19}) that appears in~\cite{bogdanov2019approximate}: by setting $\delta=0$ and $s=k$ we immediately recover the fact that when $t$ is at most some fixed constant fraction of $n$, the upper bound is governed by the $-k^2/t$ term in the exponent. Going beyond~\cite{bogdanov2019approximate}, in the case that $k$ is at least some fixed constant fraction of $n$ and $t/n\rightarrow 1$, we see that the behaviour is governed by the exponent $-(n-t)$.

Corollary~\ref{cor:main-cor-delta} is also useful in the setting that $\delta>0$. For a given $k$ and $t$, we fix a statistical distance $\varepsilon$ and then are interested in maximising $\delta$ so that symmetric $(k,\delta)$-wise indistinguishable distributions must have $t$-wise statistical distance at most $\varepsilon$. The right choices of $s$ and Corollary~\ref{cor:main-cor-delta} yield upper bounds in various regimes; in Section~\ref{s:applications} we write down explicitly the cases in which $t$ is at most a fixed constant fraction of $n$ and in which $t$ approaches $n$. These corollaries extend results in~\cite{huang2022approximate}, which only consider reconstruction by all $n$ bits and require $\delta$ to be smaller than our bounds. To gain intuition for now, we state one representative result and defer full results to Section~\ref{s:applications}.
\begin{corollary}
\label{cor:main_second_cor}
Set any $t<0.99n$ and $k\geq C\sqrt{t\log(n/\varepsilon)}$ for some sufficiently large constant $C$. If \(\mu,\nu\) are symmetric \((k,\delta)\)-wise indistinguishable distributions over \(\{0,1\}^n\), and if
\[
\delta
    \le
    \frac{\varepsilon}{2\sqrt{k+1}}
    \exp\!\left(
        -O\!\left(
            \frac{(t-k)\log(n/\varepsilon)}{k}
        \right)
    \right)
\]
then the $t$-wise statistical distance is at most $\varepsilon$.
\end{corollary}
Finally, we consider the distance between pairs of symmetric $(k,\delta)$-wise indistinguishable distributions and pairs of $(k,0)$-wise indistinguishable distributions. Our result is to improve the upper bound of Huang and Viola under the additional premise of symmetry and to prove tightness of the bound. 
For symmetric distributions \(\mu,\nu\) on \(\{0,1\}^n\), we define
\[
    \operatorname{Rep}_{n,k}(\mu,\nu)
    :=
    \inf_{\mu',\nu'}
    \max\{\Delta(\mu,\mu'),\Delta(\nu,\nu')\},
\]
where the infimum is over symmetric distributions \(\mu',\nu'\) on
\(\{0,1\}^n\) which are \((k,0)\)-wise indistinguishable.

\begin{theorem}
\label{thm:symmetric-repair-upper}
Let \(0\le k<n\), and let \(\mu,\nu\) be symmetric
\((k,\delta)\)-wise indistinguishable distributions on \(\{0,1\}^n\).
Then
\[
    \operatorname{Rep}_{n,k}(\mu,\nu)
    \le
    \delta\,2^{\,k-n\cdot I(k/n)+O(\log n)}.
\]
\end{theorem}
As an example, when \(k\) is a fixed fraction of \(n\), our upper bound is $2^{O(n)}$, whereas the general
Huang--Viola bound is \(2^{\Omega(n\log n)}\).
We also prove a matching lower bound to Theorem~\ref{thm:symmetric-repair-upper} up to polynomial factors for
sufficiently small \(\delta\).

\paragraph{Techniques and roadmap.}
A standard approach to proving indistinguishability upper bounds for
symmetric distributions is to approximate each candidate statistical
test by a low-degree polynomial. Bogdanov et
al.~\cite{bogdanov2019approximate}, for example, use classical
Chebyshev polynomials. However, such polynomials are naturally adapted to
approximation on continuous intervals and do not directly exploit the discrete structure of the Hamming-weight-based domain of approximation. In the special case
\(t=k+1\), Williamson~\cite{williamson2021sharp} used discrete
Chebyshev polynomials to obtain sharper bounds through a more
specialised analysis. 

The key method of this work is to notice and exploit the fact that local statistical tests over $t$ bits behave like strongly damped global statistical tests, and that this can be made quantitative by expressing statistical tests in a Hahn polynomial basis. The Hahn basis is naturally adapted to this because when a symmetric test that observes only \(t\) bits is lifted to a function of the global Hamming weight, its degree-\(r\) Hahn component is multiplied by an explicit damping factor, with higher-degree components more strongly damped, representing more significant ``hypergeometric smoothing''.
Truncating the resulting Hahn expansion therefore gives a controlled
low-degree approximation.
This viewpoint unifies the various arguments in the paper. For example, the main upper bound truncates the strongly damped high-degree Hahn components, while the $\delta=0$ lower bound proceeds by placing maximum weight on the least-damped Hahn component that is of sufficiently high degree to witness identical $k$-wise marginals.
We describe this technique in more detail in Section~\ref{s:hypersmooth} after preliminaries in Section~\ref{s:prelim}. In Section~\ref{s:main} we prove the main upper bound for approximate bounded indistinguishability. In Section~\ref{s:applications}, we set parameters to obtain an upper bound in the $\delta>0$ regime. Section~\ref{s:mainLBsection} contains a nearly matching lower bound when $\delta=0$ and Section~\ref{s:repairUBLB} includes an upper and lower bound related to the distance between approximately and perfectly indistinguishable pairs of distributions.

\section{Research Methods and Use of AI Tools}
\label{sec:research-methods-ai}
The author used ChatGPT, specifically GPT-5.5 Pro and GPT-5.6 Pro from OpenAI, as a
research and writing assistant during the preparation of this manuscript.
The tool was used for mathematical brainstorming, proof organization,
LaTeX drafting and revision, and exploratory comparison with related work. In particular, AI assistance was used while
developing and refining the exposition around hypergeometric smoothing,
Hahn-polynomials, approximate indistinguishability,
repair bounds, and tightness statements.

The AI-generated material was not used as an unchecked source of
mathematical claims. All theorem statements, proofs, algebraic
manipulations, asymptotic estimates, and citations appearing in the final
manuscript were reviewed and verified by the author, who takes full responsibility for the paper.

\section{Preliminaries}
\label{s:prelim}
\subsection{Rising and falling factorials}
For $x \in \mathbb{R}$ and $r \in \mathbb{Z}_{\geq 0}$, we define the falling factorial as
\[
    x^{\underline{r}}
    :=
    x(x-1)\cdots(x-r+1),
    \qquad
    x^{\underline 0}:=1
\]
and use Pochhammer’s symbol, given by
\[
(x)_r := x(x+1)\cdots(x+r-1),
\qquad
(x)_0 := 1
\]

\subsection{Symmetric functions and distributions}
\label{s:symmetry}
Let $f : \{0,1\}^n \to \mathbb{R}$ be a function.
We say that $f$ is symmetric if the output of $f$ depends only on the Hamming weight of its input.
A probability distribution $\mu$ is symmetric if the corresponding probability mass function mapping inputs to probabilities is a symmetric function.

\subsection{Characterisations of statistical distance}
We use the standard notion of statistical distance (total variation distance). For distributions over bit strings, the statistical distance $\Delta$ between $\mu$ and $\nu$ is
\[
\Delta(\mu,\nu)=\frac{1}{2}\sum_{z \in \{0,1\}^n}
\left|
\Pr_{X \sim \mu}[X = z] - \Pr_{Y \sim \nu}[Y = z]
\right|.
\]
We also consider the statistical distance between marginal distributions of $\mu$ and $\nu$. If $S \subseteq [n]$, let $\mu_S$ and
$\nu_S$ denote the corresponding marginals on the coordinates in $S$. For
$0 \leq k \leq n$, define
\[
\Delta_k(\mu,\nu)
:=
\max_{\substack{S \subseteq [n] \\ |S|=k}}
\Delta(\mu_S,\nu_S),
\]
where $\mu_S$ is the marginal distribution of $\mu$ over the indices in $S$. In this paper, $\mu$ and $\nu$ will always be symmetric, in which case $\Delta(\mu_S,\nu_S)$ is the same regardless of the choice of $S$.

When $\mu$ and $\nu$ are symmetric distributions over $\{0,1\}^n$, the marginals over any set of indices $S$ are also symmetric and the best function with which to distinguish between $\mu_S$ and $\nu_S$ may be assumed symmetric (see, for example, Facts 5 and 6 in~\cite{williamson2021sharp}). It follows that 
\begin{equation}
    \label{eq:stat_dis}
    \Delta_t(\mu,\nu)
=
\frac12
\sup_{\|f\|_\infty\le 1}
\left|
    \mathbb E[f(A_X)]-\mathbb E[f(A_Y)]
\right|,
\end{equation}
where the supremum is over functions $f:\{0,1,\dots,t\}\to[-1,1]$ and where $A_X$ (resp. $A_Y$) is the Hamming weight of the first $t$ positions of a sample $X\sim\mu$ (resp. $Y\sim\nu$).

\subsection{Bounded indistinguishability}
We say that \(\mu\) and \(\nu\) are \((k,\delta)\)-wise
indistinguishable if $\Delta_k(\mu,\nu)\le \delta$ and we write \(k\)-wise indistinguishable as shorthand for
\((k,0)\)-wise indistinguishable.

In this paper, \(\mu\) and \(\nu\) will always be symmetric, in which
case \(\Delta(\mu_S,\nu_S)\) depends only on \(|S|\), not on the
particular choice of \(S\). Thus, for symmetric distributions, one may
take \(S=\{1,\dots,k\}\) in the definition of \(\Delta_k(\mu,\nu)\).

\subsection{Hahn polynomials}
The Hahn polynomials are a family of orthogonal polynomials over a discrete set. We will use them to create approximations to statistical tests. Here we provide definitions and some basic properties of these polynomials. For an integer $n\ge 0$ and $0\le r\le n$, define
\[
    Q_r^{(n)}(x)
    :=
    Q_r(x;0,0,n),
\]
where $Q_r(x;\alpha,\beta,n)$ denotes the Hahn polynomial as described in the Digital Library of Mathematical Functions~\cite{NIST:DLMF}.  We have, also following the conventions of DLMF~\cite[§18.20(ii),(18.20.5)]{NIST:DLMF}, 
\[
    Q_r^{(n)}(x)
    =
    {}_3F_2\!\left(
        \begin{matrix}
            -r,\ r+1,\ -x\\
            1,\ -n
        \end{matrix}
        ;1
    \right)
    =
    \sum_{\ell=0}^{r}
    \frac{(-r)_\ell(r+1)_\ell(-x)_\ell}
         {(1)_\ell(-n)_\ell}
    \frac{1}{\ell!}.
\]
This finite sum is well-defined for every \(0\le r\le n\), since
\[
    (-n)_\ell=(-1)^\ell n^{\underline \ell}\neq 0
    \qquad
    \text{for }0\le \ell\le r\le n.
\]
We derive a further expression for $Q_r^{(n)}$, using:
\[
    (-r)_\ell=(-1)^\ell \frac{r!}{(r-\ell)!},
    \qquad
    (r+1)_\ell=\frac{(r+\ell)!}{r!},
    \qquad
    (-x)_\ell=(-1)^\ell x^{\underline \ell},
\]
and
\[
    (-n)_\ell=(-1)^\ell n^{\underline \ell},
    \qquad
    (1)_\ell=\ell!,
\]
to obtain
\begin{equation}
\label{eq:uniform-hahn-definition}
    Q_r^{(n)}(x)
    =
    \sum_{\ell=0}^r
    (-1)^\ell
    \binom r\ell
    \binom{r+\ell}{\ell}
    \frac{x^{\underline{\ell}}}{n^{\underline{\ell}}}.
\end{equation}

\subsection{Orthogonality}
We have that (cf.~\cite[§18.19]{NIST:DLMF})
\begin{equation}
\label{eq:hahn-orthogonality-unnormalized}
    \sum_{x=0}^n
    Q_r^{(n)}(x)Q_s^{(n)}(x)
    =
    h_{n,r}\delta_{rs},
\end{equation}
where
\begin{equation}
\label{eq:hahn-unnormalized-norm}
    h_{n,r}
    =
    \frac{(n+r+1)!(n-r)!}{(2r+1)(n!)^2}.
\end{equation}
and $\delta_{rs}$ is the Kronecker delta.

We define an inner product with:
\[
    \langle f,g\rangle_n
    :=
    \frac{1}{n+1}
    \sum_{x=0}^n f(x)g(x),
    \qquad
    \|f\|_n^2:=\langle f,f\rangle_n.
\]
With respect to this inner product, the polynomials $Q_0^{(n)},Q_1^{(n)},\dots,Q_n^{(n)}$ are mutually orthogonal. To obtain an orthonormal basis, we compute:
\begin{equation}
\label{eq:hahn-normalized-norm}
    H_{n,r}
    :=
    \|Q_r^{(n)}\|_n^2
    =
    \frac{h_{n,r}}{n+1}
    =
    \frac{(n+r+1)!(n-r)!}{(2r+1)(n!)^2(n+1)}
    =
    \frac{1}{2r+1}
    \prod_{j=0}^{r-1}
    \frac{n+j+2}{n-j},
\end{equation}
where the final equality follows from the following short calculation:
\begin{align*}
\frac{(n+r+1)!(n-r)!}{(2r+1)(n!)^2(n+1)}
&=
\frac{1}{2r+1}
\cdot
\frac{(n+r+1)!}{(n+1)n!}
\cdot
\frac{(n-r)!}{n!} \\
&=
\frac{1}{2r+1}
\cdot
\frac{(n+2)(n+3)\cdots(n+r+1)}
{n(n-1)\cdots(n-r+1)} \\
&=
\frac{1}{2r+1}
\prod_{j=0}^{r-1}
\frac{n+j+2}{n-j}.
\end{align*}
\subsection{Orthonormality}
We write $\phi_r^{(n)}$ to denote the orthonormal version of the $Q_r^{(n)}$:
\begin{equation}
\label{eq:orthonormal-hahn}
    \phi_r^{(n)}(x)
    :=
    \frac{Q_r^{(n)}(x)}{\sqrt{H_{n,r}}},
    \qquad 0\le r\le n.
\end{equation}
Then
\[
    \langle \phi_r^{(n)},\phi_s^{(n)}\rangle_n
    =
    \delta_{rs}.
\]
Since there are $n+1$ mutually orthonormal functions on the
$(n+1)$-point set $\{0,\dots,n\}$, the family $\left\{
        \phi_0^{(n)},\phi_1^{(n)},\dots,\phi_n^{(n)}
    \right\}$
is an orthonormal basis for all real-valued functions on $\{0,\dots,n\}$.
For use later, we define 
\begin{equation}
\label{eq:lambda-def}
    \lambda_{n,t,r}=\sqrt{\frac{H_{n,r}}{H_{t,r}}},
\end{equation}
which is legitimate because $H_{t,r}$ is never zero. We end this section with a short lemma for use later which is proven in the appendix.
\begin{lemma}
\label{lem:lambda-monotone}
Suppose \(0\le t<n\). Then
\[
    \lambda_{n,t,r+1}<\lambda_{n,t,r}
    \qquad
    \text{for }0\le r<t.
\]
\end{lemma}

\section{Hypergeometric smoothing}
\label{s:hypersmooth}
In this section, we define a sampling operator and a marginalisation operator and establish an adjoint relationship between them. We will use the hypergeometric distribution: let \(K_{n,t}(j,a)\) denote the hypergeometric expression
\[
    K_{n,t}(j,a)
    :=
    \frac{\binom ja\binom{n-j}{t-a}}{\binom nt},
    \qquad
    0\le j\le n,\quad 0\le a\le t.
\]

\subsection{A hypergeometric sampling operator}
Let $t\leq n$. For a function $f:\{0,1,\dots,t\}\to\mathbb R$,
define the hypergeometric sampling operator $T_{n,t}$ such that
\[
    T_{n,t}f:\{0,1,\dots,n\}\to\mathbb R
\]
where
\begin{equation}
\label{eq:hypergeometric-operator}
    (T_{n,t}f)(j)
    :=
    \sum_{a=0}^t
    f(a)K_{n,t}(j,a).
\end{equation}

The operator has the following probabilistic interpretation. Fix
\(j\in\{0,1,\dots,n\}\), let \(W\subseteq[n]\) be any fixed set of size
\(j\), choose a uniformly random subset \(S\subseteq[n]\) of size
\(t\), and define
\[
    A:=|S\cap W|.
\]
Then \(A\) has the hypergeometric distribution $\Pr[A=a]=K_{n,t}(j,a)$,
and hence
\begin{equation}
\label{eq:hypergeometric-operator-hypergeometric}
    (T_{n,t}f)(j)
    =
    \mathbb E[f(A)].
\end{equation}
An equivalent perspective is to uniformly permute the coordinates of the string
\(1^j0^{n-j}\) and let \(A\) be the Hamming weight of its first \(t\)
coordinates.

More specifically, let \(X\) be sampled from a symmetric distribution on
\(\{0,1\}^n\), and let
\[
    A_X:=\sum_{i=1}^t X_i.
\]
Conditional on \(|X|=j\), the support of \(X\) is a uniformly random
\(j\)-subset of \([n]\). Therefore
\begin{equation}
    \label{eq:hypergeometric-operator-hypergeometric}
    \mathbb E[f(A_X)\mid |X|=j]
    =
    (T_{n,t}f)(j)
\end{equation}
for every \(j\) in the support of \(|X|\). By the law of total
expectation,
\begin{equation}
\label{eq:sampling-operator-expectation}
    \mathbb E[f(A_X)]
    =
    \mathbb E\!\left[(T_{n,t}f)(|X|)\right].
\end{equation}

We establish a lemma about the expected value of $A$ in the experiment described above. In a subsequent lemma, we will use this to establish a relationship showing how the operator $T_{n,t}$ acts on the polynomials $\phi_r^{(n)}$. 
\begin{lemma}
\label{lem:hypergeometric-factorial-moments}
Let \(0\le t\le n\).  Fix \(j\in\{0,1,\dots,n\}\). Let \(W\subseteq[n]\) be a fixed set of size \(j\), choose
a subset \(S\subseteq[n]\) of size \(t\) uniformly at random, and define
\[
    A:=|S\cap W|.
\]
Then, for every \(0\le \ell\le t\),
\[
    \mathbb E\left[A^{\underline{\ell}}\right]
    =
    \frac{t^{\underline{\ell}}}{n^{\underline{\ell}}}
    j^{\underline{\ell}}.
\]
\end{lemma}

\begin{proof}
The case \(\ell=0\) is immediate because both sides equal 1. We now assume \(\ell\ge 1\).

The random variable \(\binom A\ell\) counts size \(\ell\) subsets of
\(S\cap W\).  For each \(U\subseteq W\) with \(|U|=\ell\), let $\mathbf 1_{U\subseteq S}$ be an indicator random variable for the event that $U\subseteq S$.
Then
\[
    \binom A\ell
    =
    \sum_{\substack{U\subseteq W\\ |U|=\ell}} \mathbf 1_{U\subseteq S}.
\]
Taking expectations and using linearity,
\[
    \mathbb E\left[\binom A\ell\right]
    =
    \sum_{\substack{U\subseteq W\\ |U|=\ell}}
    \Pr[U\subseteq S].
\]
For a fixed \(\ell\)-set \(U\), since \(S\) is a uniformly random
\(t\)-subset of \([n]\),
\[
    \Pr[U\subseteq S]
    =
    \frac{\binom{n-\ell}{t-\ell}}{\binom nt}
    =
    \frac{t^{\underline \ell}}{n^{\underline \ell}}.
\]
There are \(\binom j\ell\) choices of \(U\), so
\[
    \mathbb E\left[\binom A\ell\right]
    =
    \binom j\ell
    \frac{t^{\underline \ell}}{n^{\underline \ell}}.
\]
Multiplying both sides by \(\ell!\) and using the identity $m^{\underline \ell}=\ell!\binom m\ell$, yields that $\mathbb{E}\left[A^{\underline{\ell}}\right]=j^{\underline \ell}\frac{t^{\underline \ell}}{n^{\underline \ell}}$.
\end{proof}
We are now in a position to show how operator $T_{n,t}$ acts on the orthonormal Hahn polynomials.
\begin{lemma}
\label{lem:hahn-intertwining}
For every $0\le r\le t$, 
\[
T_{n,t}\phi_r^{(t)}=\lambda_{n,t,r}\phi_r^{(n)},
\]
where the $\lambda_{n,t,r}$ are defined in Equation~\ref{eq:lambda-def}.
\end{lemma}

\begin{proof}
Fix \(j\in\{0,\dots,n\}\). Let \(W\subseteq[n]\) be any fixed set
of size \(j\), choose a uniformly random \(t\)-subset
\(S\subseteq[n]\), and define $A:=|S\cap W|$.
By the definition of \(T_{n,t}\), \eqref{eq:hypergeometric-operator-hypergeometric}, \eqref{eq:uniform-hahn-definition}, and
Lemma~\ref{lem:hypergeometric-factorial-moments}, we obtain
\[
\begin{aligned}
(T_{n,t}Q_r^{(t)})(j)
&=
\mathbb E\!\left[Q_r^{(t)}(A)\right] \\
&=
\mathbb E\!\left[
    \sum_{\ell=0}^r
    (-1)^\ell
    \binom r\ell
    \binom{r+\ell}{\ell}
    \frac{A^{\underline{\ell}}}{t^{\underline{\ell}}}
\right] \\
&=
\sum_{\ell=0}^r
(-1)^\ell
\binom r\ell
\binom{r+\ell}{\ell}
\frac{j^{\underline{\ell}}}{n^{\underline{\ell}}} \\
&=
Q_r^{(n)}(j).
\end{aligned}
\]
Since this holds for every \(j\in\{0,\dots,n\}\), we have $T_{n,t}Q_r^{(t)}=Q_r^{(n)}$.
Then, linearity of the operator $T_{n,t}$ and dividing by the normalizing
constants gives
\[
    T_{n,t}\phi_r^{(t)}=\frac{Q_r^{(n)}}{\sqrt{H_{t,r}}}=\sqrt{\frac{H_{n,r}}{H_{t,r}}}=\frac{Q_r^{(n)}}{\sqrt{H_{n,r}}}=\lambda_{n,t,r}\phi_r^{(n)}.
\]
\end{proof}

\subsection{A marginalisation operator and adjointness}
For \(0\le t\le n\), define the marginalisation operator
\(M_{n\to t}\) on signed measures \(\sigma\) on
\(\{0,\dots,n\}\) by
\begin{equation}
\label{eq:marginalisation-operator}
    (M_{n\to t}\sigma)(a)
    :=
    \sum_{j=0}^n
    K_{n,t}(j,a)
    \sigma(j),
    \qquad 0\le a\le t.
\end{equation}
Thus, if \(\sigma\) is the Hamming-weight distribution of a
symmetric distribution on \(\{0,1\}^n\), then
\(M_{n\to t}\sigma\) is the Hamming-weight distribution of its
marginal on any \(t\) coordinates.

The operators $T_{n,t}$ and $M_{n\to t}$ are adjoint in the following sense:
\begin{equation}
\label{eq:T-M-adjoint-measures}
    \sum_{a=0}^t f(a)(M_{n\to t}\sigma)(a)
    =
    \sum_{j=0}^n (T_{n,t}f)(j)\sigma(j).
\end{equation}
Indeed, for every function
\(f:\{0,\dots,t\}\to\mathbb R\) and every signed measure
\(\sigma\) on \(\{0,\dots,n\}\).
\[
\begin{aligned}
    \sum_{a=0}^t f(a)(M_{n\to t}\sigma)(a)
    &=
    \sum_{a=0}^t
    f(a)
    \sum_{j=0}^n K_{n,t}(j,a)\sigma(j)          \\
    &=
    \sum_{j=0}^n
    \left(
        \sum_{a=0}^t f(a)K_{n,t}(j,a)
    \right)
    \sigma(j)                                  \\
    &=
    \sum_{j=0}^n (T_{n,t}f)(j)\sigma(j).
\end{aligned}
\]
We now turn to a marginalisation operator analog of Lemma~\ref{lem:hahn-intertwining}. For \(0\le r\le t\), define the signed measure
\begin{equation}
\label{eq:hahn-signed-measure}
    \theta_{t,r}(a)
    :=
    \frac{1}{t+1}\phi_r^{(t)}(a),
    \qquad 0\le a\le t.
\end{equation}

\begin{lemma}
\label{lem:hahn-measure-intertwining}
Let \(0\le t\le n\) and \(0\le r\le n\). Then
\[
    M_{n\to t}\theta_{n,r}
    =
    \begin{cases}
        \lambda_{n,t,r}\theta_{t,r},
            & 0\le r\le t,\\[1ex]
        0,
            & t<r\le n.
    \end{cases}
\]
\end{lemma}

\begin{proof}
Let \(f:\{0,\dots,t\}\to\mathbb R\) be arbitrary. By
Equation~\eqref{eq:T-M-adjoint-measures} and the definition of
\(\theta_{n,r}\),
\[
\begin{aligned}
    \sum_{a=0}^t
    f(a)(M_{n\to t}\theta_{n,r})(a)
    &=
    \sum_{j=0}^n
    (T_{n,t}f)(j)\theta_{n,r}(j)\\
    &=
    \left\langle
        T_{n,t}f,\phi_r^{(n)}
    \right\rangle_n.
\end{aligned}
\]

Expand \(f\) in the orthonormal Hahn basis on
\(\{0,\dots,t\}\):
\[
    f
    =
    \sum_{q=0}^t b_q\phi_q^{(t)},
    \qquad
    b_q
    =
    \left\langle f,\phi_q^{(t)}\right\rangle_t.
\]
By Lemma~\ref{lem:hahn-intertwining},
\[
    T_{n,t}f
    =
    \sum_{q=0}^t
    b_q\lambda_{n,t,q}\phi_q^{(n)}.
\]
Orthonormality on \(\{0,\dots,n\}\) therefore gives
\[
    \left\langle
        T_{n,t}f,\phi_r^{(n)}
    \right\rangle_n
    =
    \begin{cases}
        b_r\lambda_{n,t,r},
            & r\le t,\\
        0,
            & r>t.
    \end{cases}
\]
When \(r\le t\), the definition of \(b_r\) and
Equation~\eqref{eq:hahn-signed-measure} give
\[
    b_r
    =
    \frac{1}{t+1}
    \sum_{a=0}^t f(a)\phi_r^{(t)}(a)
    =
    \sum_{a=0}^t f(a)\theta_{t,r}(a).
\]
Consequently,
\[
    \sum_{a=0}^t
    f(a)(M_{n\to t}\theta_{n,r})(a)
    =
    \begin{cases}
        \displaystyle
        \sum_{a=0}^t
        f(a)\lambda_{n,t,r}\theta_{t,r}(a),
            & r\le t,\\[3ex]
        0,
            & r>t.
    \end{cases}
\]
Since this identity holds for every \(f\), the claimed equality of
signed measures follows.
\end{proof}

\section{Proof of Theorem~\ref{thm:main-exact-H-form}}
\label{s:main}
In this section, we prove our main upper bound for approximate bounded indistinguishability. In particular, we prove the following sharper restatement of Theorem~\ref{thm:main-exact-H-form}. Theorem~\ref{thm:main-exact-H-form} follows from our restatement via an easy application of Stirling's approximation. 

\begin{theorem}[Restatement of Theorem~\ref{thm:main-exact-H-form}]
\label{thm:mainmain}
    Let \(\mu,\nu\) be symmetric distributions over \(\{0,1\}^n\) which are
\((k,\delta)\)-wise indistinguishable.  Then, for integers $0\le s\le k<t<n$,
one has
\[
\Delta_t(\mu,\nu)\leq \sqrt{n+1}\cdot\lambda_{n,t,s+1}+\delta\,\sqrt{k+1}\cdot \lambda^{-1}_{t,k,s}.
\]
\end{theorem}
In a nutshell, our technique is to consider statistical tests
that observe \(t>k\) bits. By symmetry, it is enough to consider tests
\(f:\{0,1,\dots,t\}\to[-1,1]\) depending only on the observed Hamming
weight. Against a symmetric distribution, the
expected value of such a test is captured by its smoothed global version
\(T_{n,t}f\), evaluated at the total Hamming weight. In Lemma~\ref{lem:hahn-cutoff-decomposition}, we find a test $h$ that observes $\ell<t$ bits, is of degree $s\leq\ell$, and has a smoothed global version that closely approximates \(T_{n,t}f\). As such, the statistical distinguishing power of $f$ is limited by this approximation error (quantified in Lemma~\ref{lem:hahn-cutoff-decomposition}) and the distinguishing power of $h$ (upper bounded by Lemma~\ref{lem:approx-controls-low-degree} in view of the locality of $h$ and $(k,\delta)$-wise indistinguishability). 
We pick $h$ so that its smoothed version is the degree $s$ Hahn polynomial truncation of \(T_{n,t}f\). There are multiple candidates for $h$ that all smooth to the target Hahn truncation, in particular, we have the freedom to choose $\ell$ and construct a corresponding $h$ observing $\ell$ bits. We then write $h$ in the Hahn basis $\{\phi^{(\ell)}_r\}_{0\leq r\leq \ell}$ and can do so as long as the basis is of high enough degree ($\ell\geq s$). We will also require that $\ell\leq k$ so that $(k,\delta)$-wise indistinguishability limits the distinguishing power of $h$. The proof of Theorem~\ref{thm:mainmain} puts these pieces together and optimises over parameters.

\subsection{Approximating general tests}
We establish an approximation-theoretic result before translating into the language of approximate bounded indistinguishability.
\begin{lemma}
\label{lem:hahn-cutoff-decomposition}
Let \(0\le s\leq\ell<t<n\), and let quantities $\lambda_{\cdot,\cdot,\cdot}$ be as defined by Equation~\ref{eq:lambda-def}. Take any
\[
    f:\{0,1,\dots,t\}\to[-1,1].
\]
Then there exists a degree $\leq s$ polynomial $h_{f,\ell}:\{0,1,\dots,\ell\}\to\mathbb R$
such that $\|h_{f,\ell}\|_\infty\le\sqrt{\ell+1}\cdot \lambda^{-1}_{t,\ell,s}$ and
\begin{equation}
\label{eq:hahn-cutoff-remainder-bound}
    \|T_{n,t}f-T_{n,\ell}h_{f,\ell}\|_\infty
    \le
    \sqrt{n+1}\cdot\lambda_{n,t,s+1}.
\end{equation}
\end{lemma}

\begin{proof}
Expand \(f\) in the orthonormal Hahn basis on \(\{0,\dots,t\}\):
\[
    f
    =
    \sum_{r=0}^t a_r\phi_r^{(t)},
\]
and recall that from Lemma~\ref{lem:hahn-intertwining} we have $T_{n,t}\phi_r^{(t)}=\lambda_{n,t,r}\phi_r^{(n)}$ and the implication $T_{n,t}f=\sum_{r=0}^ta_r\lambda_{n,t,r}\phi_r^{(n)}$. 
Since \(|f|\le 1\), and by orthogonality of the $\phi_r^{(t)}$,
\[
    \sum_{r=0}^t a_r^2
    =
    \|f\|_t^2
    =
    \frac1{t+1}\sum_{a=0}^t f(a)^2
    \le 1.
\]
Then, we define
\[
    h_{f,\ell}
    :=
    \sum_{r=0}^s
    a_r
    \frac{\lambda_{n,t,r}}{\lambda_{n,\ell,r}}
    \phi_r^{(\ell)},
\]
which is legitimate because $\lambda_{n,\ell,r}$ is nonzero as it is a ratio of positive values. 
Linearity of the operator $T$ gives
\[
T_{n,\ell}h_{f,\ell}=\sum_{r=0}^sa_r\frac{\lambda_{n,t,r}}{\lambda_{n,\ell,r}}
    T_{n,\ell}\phi_r^{(\ell)}=\sum_{r=0}^sa_r
    \frac{\lambda_{n,t,r}}{\lambda_{n,\ell,r}}
    \lambda_{n,\ell,r}\phi_r^{(n)}=\sum_{r=0}^sa_r\lambda_{n,t,r}\phi_r^{(n)},
\]
where the result is just the degree $\leq s$ portion of $T_{n,t}f$.
The remainder is
\[
    R_f
    :=
    T_{n,t}f-T_{n,\ell}h_{f,\ell}
    =
    \sum_{r=s+1}^t
    a_r\lambda_{n,t,r}\phi_r^{(n)}.
\]
By orthonormality of the Hahn polynomials on \(\{0,\dots,n\}\), $\|R_f\|_n^2=\sum_{r=s+1}^ta_r^2\lambda_{n,t,r}^2$.
For every $r>s$, Lemma~\ref{lem:lambda-monotone} gives that $\lambda_{n,t,r}\leq\lambda_{n,t,s+1}$ and therefore
\[
    \|R_f\|_n^2
    \le
    \lambda^2_{n,t,s+1}
    \sum_{r=s+1}^t a_r^2
    \le
    \lambda^2_{n,t,s+1}.
\]
Hence $\|R_f\|_n\le \lambda_{n,t,s+1}$. Every function $g$ on \(\{0,\dots,n\}\) satisfies
$\|g\|_\infty^2\le\sum_{x=0}^n g(x)^2=(n+1)\|g\|_n^2$,
and thus $\|g\|_\infty\le \sqrt{n+1}\,\|g\|_n$.
From this, we obtain
\[
    \|T_{n,t}f-T_{n,\ell}h_{f,\ell}\|_\infty
    =
    \|R_f\|_\infty
    \le
    \sqrt{n+1}\cdot\lambda_{n,t,s+1}.
\]
Finally, we compute an upper bound on $||h_{f,\ell}||_{\infty}$. Directly from the definitions of $\lambda_{n,m,r}$ and $h_{f,\ell}$, we have:
\[
    h_{f,\ell}
    =
    \sum_{r=0}^s
    a_r
    \sqrt{\frac{H_{\ell,r}}{H_{t,r}}}
    \phi_r^{(\ell)}.
\]
By orthonormality on \(\{0,\dots,\ell\}\),
\[
\begin{aligned}
    \|h_{f,\ell}\|_\ell^2
    =
    \sum_{r=0}^s
    a_r^2
    \frac{H_{\ell,r}}{H_{t,r}}
    &\le
    \left(
        \max_{0\le r\le s}
        \frac{H_{\ell,r}}{H_{t,r}}
    \right)
    \sum_{r=0}^s a_r^2                                      \\
    &\le
    \max_{0\le r\le s}
    \frac{H_{\ell,r}}{H_{t,r}} \\
    &=\max_{0\le r\le s}
    \lambda^{-2}_{t,\ell,r}, \\
    &=
    \lambda^{-2}_{t,\ell,s},
\end{aligned}
\]
where the final inequality is from Lemma~\ref{lem:lambda-monotone}.
Since $\|h_{f,\ell}\|_\infty\le \sqrt{\ell+1}\,\|h_{f,\ell}\|_\ell$,
we have
\[
    \|h_{f,\ell}\|_\infty
    \le
    \sqrt{\ell+1}\cdot \lambda^{-1}_{t,\ell,s}.
\]
\end{proof}

\subsection{Approximate bounded indistinguishability}
We begin by stating a simple lemma about the distinguishing power of tests that observe $\ell\leq k$ bits. 

\begin{lemma}
\label{lem:approx-controls-low-degree}
Let \(\mu,\nu\) be symmetric distributions over \(\{0,1\}^n\) which are
\((k,\delta)\)-wise indistinguishable. Let \(0\le \ell\le k\), and let
\[
    h:\{0,1,\dots,\ell\}\to\mathbb R.
\]
Then
\[
    \left|
        \mathbb E_{X\sim\mu}[(T_{n,\ell}h)(|X|)]
        -
        \mathbb E_{Y\sim\nu}[(T_{n,\ell}h)(|Y|)]
    \right|
    \le
    2\delta\|h\|_\infty.
\]
\end{lemma}
A proof of this lemma appears in the Appendix.
We now combine Lemmas~\ref{lem:hahn-cutoff-decomposition} and~\ref{lem:approx-controls-low-degree} to prove Theorem~\ref{thm:mainmain}.
\begin{proof}[Proof of Theorem~\ref{thm:mainmain}]
Let
\[
    A_X:=\sum_{i=1}^t X_i,
    \qquad
    A_Y:=\sum_{i=1}^t Y_i.
\]
By symmetry and Equation~\ref{eq:stat_dis},
\[
    \Delta_t(\mu,\nu)
    =
    \frac12
    \sup_{\|f\|_\infty\le 1}
    \left|
        \mathbb E[f(A_X)]-\mathbb E[f(A_Y)]
    \right|,
\]
where the supremum is over functions $f:\{0,1,\dots,t\}\to[-1,1]$.
Fix such an \(f\).  For some $\ell\in\{s,s+1,...,k\}$ to be set later, apply Lemma~\ref{lem:hahn-cutoff-decomposition} to
obtain \(h_{f,\ell}\) of degree $\leq s$ satisfying the conditions of that lemma. From Equation~\ref{eq:sampling-operator-expectation} we have $\mathbb E[f(A_X)]=\mathbb E[(T_{n,t}f)(|X|)]$ and $\mathbb E[f(A_Y)]=\mathbb E[(T_{n,t}f)(|Y|)]$.
Thus,
\[
\begin{aligned}
    \mathbb E[f(A_X)]-\mathbb E[f(A_Y)]
    &=
    \mathbb E[(T_{n,t}f)(|X|)]
    -
    \mathbb E[(T_{n,t}f)(|Y|)]                                  \\
    &=
    \Big(
        \mathbb E[(T_{n,t}f-T_{n,\ell}h_{f,\ell})(|X|)]
        -
        \mathbb E[(T_{n,t}f-T_{n,\ell}h_{f,\ell})(|Y|)]
    \Big)                                                       \\
    &\qquad
    +
    \Big(
        \mathbb E[(T_{n,\ell}h_{f,\ell})(|X|)]
        -
        \mathbb E[(T_{n,\ell}h_{f,\ell})(|Y|)]
    \Big).
\end{aligned}
\]
The first term in the sum is bounded as follows:
\[
\begin{aligned}
    &
    \left|
        \mathbb E[(T_{n,t}f-T_{n,\ell}h_{f,\ell})(|X|)]
        -
        \mathbb E[(T_{n,t}f-T_{n,\ell}h_{f,\ell})(|Y|)]
    \right|                                                     \\
    &\qquad\le
    2\|T_{n,t}f-T_{n,\ell}h_{f,\ell}\|_\infty                                    \\
    &\qquad\le
    2\sqrt{n+1}\cdot\lambda_{n,t,s+1},
\end{aligned}
\]
where the last step uses Lemma~\ref{lem:hahn-cutoff-decomposition}.
For the second term, use Lemma~\ref{lem:approx-controls-low-degree}.  Since \(\ell\le k\),
\[
\begin{aligned}
    \left|
        \mathbb E[(T_{n,\ell}h_{f,\ell})(|X|)]-\mathbb E[(T_{n,\ell}h_{f,\ell})(|Y|)]
    \right|                                                  
    &\le
    2\delta\|h_{f,\ell}\|_\infty                                      \\
    &\le
    2\delta\,\sqrt{\ell+1}\cdot \lambda^{-1}_{t,\ell,s}.
\end{aligned}
\]
Combining the two estimates, taking the supremum over all \(\|f\|_\infty\le 1\), and dividing by two (matching the constant in Equation~\ref{eq:stat_dis}) gives that 
\begin{equation}
    \label{eq:main_with_lambda}
    \Delta_t(\mu,\nu)\leq \sqrt{n+1}\cdot\lambda_{n,t,s+1}+\delta\,\sqrt{\ell+1}\cdot \lambda^{-1}_{t,\ell,s},
\end{equation}
from which the proof follows by choosing $\ell=k$.
\end{proof}

We note that we are free to set $\ell$ arbitrarily up to $k$. Very modest gains could be made from optimising our selection.

Theorem~\ref{thm:main-exact-H-form}, as stated in the introduction, follows immediately by the definition of $\lambda_{n,t,s+1}$ and $\lambda_{t,k,s}$ as given by Equations~\ref{eq:hahn-normalized-norm} and~\ref{eq:lambda-def}.

\subsubsection{Proof of Corollary~\ref{cor:main-cor-delta}}
\label{s:subsubsectionCor}
In this section, we use upper and lower bounds on the scaling factors $\lambda_{n,t,r}$, which are proved in the Appendix.
\begin{lemma}
\label{lem:nonasymptotic-damping}
Let \(0\le r\le t<n\). Then
\begin{equation}
\label{eq:single-r-damping}
    \lambda_{n,t,r}
    \le
    \exp\!\left(
        -
        \frac{(n-t)r(r+1)}
             {2n(t+r+1)}
    \right).
\end{equation}
\end{lemma}

\begin{lemma}
\label{lem:lambda-inverse-bound}
Let \(0\le r\le t< n\). Then
\[
    \lambda_{n,t,r}^{-1}
    \le
    \exp\!\left(
        \frac{(n-t)r(r+1)}
             {2n(t-r+1)}
    \right).
\]
\end{lemma}

From here Corollary~\ref{cor:main-cor-delta} has a very short proof:
\begin{proof} (of Corollary~\ref{cor:main-cor-delta})
    Start from the bound of Theorem~\ref{thm:mainmain}, and bound $\lambda_{n,t,s+1}$ with Lemma~\ref{lem:nonasymptotic-damping} setting $r=s+1$ and using the basic estimates: $(s+1)(s+2)\geq s^2$ and $2n(t+s+2)\leq 2n(2t+1)\leq 6nt$. To estimate $\lambda^{-1}_{t,k,s}$, apply Lemma~\ref{lem:lambda-inverse-bound} setting $n,t,r$ in the lemma to $t,k,s$, respectively. 
\end{proof}

\section{Upper bounds for the $\delta>0$ regime}
\label{s:applications}

We use Corollary~\ref{cor:main-cor-delta} to develop our upper bounds. We start with a general corollary and then consider two regimes: when $t$ is at most a fixed constant fraction of $n$ and when $t$ approaches $n$.

\begin{corollary}
\label{cor:finite-epsilon-criterion}
Let \(0<\varepsilon<1\), and let \(0\le k<t<n\). Define
\[
    s
    :=
    \left\lceil
        \sqrt{
            \frac{6nt}{n-t}\cdot
            \log\left(\frac{2\sqrt{n+1}}{\varepsilon}\right)
        }
    \right\rceil,
\]
and suppose $s\leq k$. Then, if \(\mu,\nu\) are symmetric \((k,\delta)\)-wise indistinguishable
distributions on \(\{0,1\}^n\), and
\[
    \delta
    \le
    \frac{\varepsilon}{2}
    \exp\!\left(
        -\frac12\log(k+1)
        -
        \frac{(t-k)s(s+1)}
             {2t(k-s+1)}
    \right),
\]
then $\Delta_t(\mu,\nu)\le \varepsilon$.
\end{corollary}

\begin{proof}
Under the premise that $s\leq k$, we may apply Corollary~\ref{cor:main-cor-delta} to obtain
\[
    \Delta_t(\mu,\nu)
    \le
    \exp\!\left(
        \frac12\log(n+1)
        -
        \frac{(n-t)s^2}{6nt}
    \right)
    +
    \delta
    \exp\!\left(
        \frac12\log(k+1)
        +
        \frac{(t-k)s(s+1)}
             {2t(k-s+1)}
    \right).
\]
From our setting of the parameter $s$,
\[
    \frac{(n-t)s^2}{6nt}
    \ge
    \log\left(\frac{2\sqrt{n+1}}{\varepsilon}\right),
\]
and therefore
\[
    \exp\!\left(
        \frac12\log(n+1)
        -
        \frac{(n-t)s^2}{6nt}
    \right)
    \le
    \frac{\varepsilon}{2}.
\]
The assumed upper bound on \(\delta\) gives
\[
    \delta
    \exp\!\left(
        \frac12\log(k+1)
        +
        \frac{(t-k)s(s+1)}
             {2t(k-s+1)}
    \right)
    \le
    \frac{\varepsilon}{2}.
\]
Combining the two estimates gives $\Delta_t(\mu,\nu)\le \varepsilon$.
\end{proof}

\subsection{Common parameter regimes}
\subsubsection{$t$ bounded away from $n$.}
Suppose \(t\le (1-c)n\) for some fixed \(0<c<1\). Then
\[
    s=O_c\!\left(\sqrt{t\log(n/\varepsilon)}\right).
\]
Consequently, if
\[
    k\ge C\sqrt{t\log(n/\varepsilon)}
\]
for a sufficiently large constant \(C=C(c)\), then \(s\le k\) and
\[
    \frac{(t-k)s(s+1)}
         {2t(k-s+1)}
    =
    O_c\!\left(
        \frac{(t-k)\log(n/\varepsilon)}{k}
    \right).
\]
It follows that there exists a constant \(K=K(c)>0\) such that it is
enough to assume
\[
    \delta
    \le
    \frac{\varepsilon}{2\sqrt{k+1}}
    \exp\!\left(
        -K
        \frac{(t-k)\log(n/\varepsilon)}{k}
    \right)
\]
to conclude that \(\Delta_t(\mu,\nu)\le\varepsilon\).

In particular, if $k=o(n)$, then for some $k=\omega\!\left(\sqrt{t\log n}\right)$ and $\delta=n^{-\omega(t/k)}$, then every pair of symmetric \((k,\delta)\)-wise indistinguishable distributions satisfies $\Delta_t(\mu,\nu)=n^{-\omega(1)}$. 
In the special case \(k=\Omega(t)\), if $t=\omega(\log n)$ and $\delta=n^{-\omega(1)}$, then $\Delta_t(\mu,\nu)=n^{-\omega(1)}$.

\subsubsection{$t$ approaching $n$.}
Suppose \(t=n-q\) where $q=o(n)$. Then $s=O\!\left(n\sqrt{\frac{\log(n/\varepsilon)}{q}}\right)$ and if
\[
    k
    \ge
    C \cdot n\sqrt{\frac{\log(n/\varepsilon)}{q}}
\]
for a sufficiently large constant \(C\),
\[
    \frac{(t-k)s(s+1)}
         {2t(k-s+1)}
    =
    O\!\left(
        \frac{(t-k)n^2\log(n/\varepsilon)}
             {tqk}
    \right).
\]
Consequently, it is enough to assume
\[
    \delta
    \le
    \frac{\varepsilon}{2\sqrt{k+1}}
    \exp\!\left(
        -O\!\left(
            \frac{(t-k)n^2\log(n/\varepsilon)}
                 {tqk}
        \right)
    \right)
\]
to conclude $\Delta_t(\mu,\nu)\le \varepsilon$.
A coarser sufficient condition via the lower bound on $k$ is
\[
    \delta
    \le
    \frac{\varepsilon}{2\sqrt{k+1}}
    \exp\!\left(
        -O\!\left(
            n\sqrt{\frac{\log(n/\varepsilon)}{q}}
        \right)
    \right).
\]

\section{A matching lower bound when $\delta=0$}
\label{s:mainLBsection}
Our goal in this section is to prove the following lower bound.
\begin{theorem}
\label{thm:exact-case-tightness}
Let \(0\le k<t\le n\). There exist symmetric \((k,0)\)-wise indistinguishable distributions
\(\mu,\nu\) on \(\{0,1\}^n\) such that
\[
    \Delta_t(\mu,\nu)
    \ge
    \frac{1}{\sqrt{t+1}}
    \lambda_{n,t,k+1}.
\]
\end{theorem}
In contrast to~\cite{bogdanov2019approximate} and~\cite{williamson2021sharp}, we give a direct proof that does not appeal to linear programming duality but instead parallels the upper bound's use of Hahn polynomials. The role of the Hahn polynomials in both the upper and lower bounds is
not accidental. In our construction, we take the positive and negative parts of the degree-\(r\) Hahn signed measure on \(\{0,\dots,n\}\) and observe that when $r>k$, marginalisation to \(k\) coordinates certifies $(k,0)$-wise indistinguishability whereas
marginalisation to \(t>r\) coordinates produces a related signed measure scaled by \(\lambda_{n,t,r}\). To maximise the statistical distance between the two distributions, we set $r=k+1$ so as to maximise \(\lambda_{n,t,r}\) while retaining $(k,0)$-wise indistinguishability.

\begin{proof}[Proof of Theorem~\ref{thm:exact-case-tightness}]
Let $r:=k+1$. From Equations~\ref{eq:hahn-orthogonality-unnormalized} and~\ref{eq:hahn-unnormalized-norm}, we have that $\theta_{n,r}$ is nonzero. Moreover, since \(r\ge 1\), the Hahn polynomial \(\phi_r^{(n)}\) is orthogonal to
the constant function. Therefore the signed measure
\(\theta_{n,r}\) as defined in Equation~\ref{eq:hahn-signed-measure} has total mass zero, in particular:
\[
\begin{aligned}
    \sum_{j=0}^n \theta_{n,r}(j)
    &=
    \frac{1}{n+1}
    \sum_{j=0}^n \phi_r^{(n)}(j) \\
    &=
    \left\langle
        1,\phi_r^{(n)}
    \right\rangle_n \\
    &=
    \left\langle
        Q_0^{(n)},\phi_r^{(n)}
    \right\rangle_n \\
    &=0.
\end{aligned}
\]

Define the positive and negative parts of \(\theta_{n,r}\) by
\[
    \theta_{n,r}^{+}(j)
    :=
    \max\{\theta_{n,r}(j),0\},
    \qquad
    \theta_{n,r}^{-}(j)
    :=
    \max\{-\theta_{n,r}(j),0\},
\]
and let
\[
    L
    :=
    \sum_{j=0}^n
    |\theta_{n,r}(j)|.
\]
Since \(\theta_{n,r}\) is nonzero and has total mass zero,
\[
    \sum_{j=0}^n \theta_{n,r}^{+}(j)
    =
    \sum_{j=0}^n \theta_{n,r}^{-}(j)
    =
    \frac{L}{2}.
\]
Consequently,
\[
    \alpha(j)
    :=
    \frac{2\theta_{n,r}^{+}(j)}{L},
    \qquad
    \beta(j)
    :=
    \frac{2\theta_{n,r}^{-}(j)}{L},
    \qquad 0\le j\le n,
\]
are probability measures on \(\{0,\dots,n\}\).

Their signed Hamming-weight difference is
\[
    \alpha-\beta
    =
    \frac{2}{L}\theta_{n,r}.
\]

Since \(r=k+1>k\), Lemma~\ref{lem:hahn-measure-intertwining} gives
\begin{equation}
\label{eq:margAB}
    \begin{aligned}
    M_{n\to k}(\alpha-\beta)
    &=
    \frac{2}{L}
    M_{n\to k}\theta_{n,r} \\
    &=0.
\end{aligned}
\end{equation}

Let \(\mu,\nu\) be the corresponding symmetric distributions on
\(\{0,1\}^n\). Explicitly, for \(x\in\{0,1\}^n\),
\[
    \mu(x)
    :=
    \frac{\alpha(|x|)}{\binom{n}{|x|}},
    \qquad
    \nu(x)
    :=
    \frac{\beta(|x|)}{\binom{n}{|x|}}.
\]

By Equation~\ref{eq:margAB}, the \(k\)-coordinate marginals of \(\mu\) and \(\nu\) have the
same Hamming-weight distribution. Since these marginals are symmetric,
they are identical. Hence \(\mu,\nu\) are
\((k,0)\)-wise indistinguishable.

On the other hand, \(r=k+1\le t\), so another application of
Lemma~\ref{lem:hahn-measure-intertwining} gives
\[
\begin{aligned}
    M_{n\to t}(\alpha-\beta)
    &=
    \frac{2}{L}
    M_{n\to t}\theta_{n,r} \\
    &=
    \frac{2\lambda_{n,t,r}}{L}
    \theta_{t,r}.
\end{aligned}
\]
The \(t\)-coordinate marginals are symmetric, so their statistical
distance equals the statistical distance between their Hamming-weight
distributions. Therefore
\begin{align}
    \Delta_t(\mu,\nu)
    &=
    \frac12
    \sum_{a=0}^t
    \left|
        \bigl(M_{n\to t}(\alpha-\beta)\bigr)(a)
    \right| \notag\\
    &=
    \frac{\lambda_{n,t,r}}{L}
    \sum_{a=0}^t
    |\theta_{t,r}(a)|.
    \label{eq:exact-lower-construction}
\end{align}

We now bound the numerator and denominator in
Equation~\eqref{eq:exact-lower-construction}. The fact that the 1-norm of a vector is upper bounded by the square root of its length times the 2-norm, and the
normalization of \(\phi_r^{(n)}\),
\[
L=\frac{1}{n+1}\sum_{j=0}^n|\phi_r^{(n)}(j)|\le
    \left(
        \frac{1}{n+1}
        \sum_{j=0}^n
        \bigl(\phi_r^{(n)}(j)\bigr)^2
    \right)^{1/2}=1.
\]
Conversely, using normalization and the fact that the 1-norm is at least the 2-norm,
\[
    \sum_{a=0}^t
    |\theta_{t,r}(a)|
    =
    \frac{1}{t+1}
    \sum_{a=0}^t
    |\phi_r^{(t)}(a)| 
    \ge
    \frac{1}{t+1}
    \left(
        \sum_{a=0}^t
        \bigl(\phi_r^{(t)}(a)\bigr)^2
    \right)^{1/2} 
    =
    \frac{1}{\sqrt{t+1}}.
\]
Substituting these estimates into
Equation~\eqref{eq:exact-lower-construction} gives $\Delta_t(\mu,\nu)\ge\frac{\lambda_{n,t,k+1}}{\sqrt{t+1}}$,
as required.
\end{proof}

When \(t<n\), we compare to Theorem~\ref{thm:mainmain} with
\(\delta=0\). By Lemma~\ref{lem:lambda-monotone}, the strongest possible upper bound is when $s=k$, and we obtain an upper bound of $\sqrt{n+1}\cdot\lambda_{n,t,k+1}$, matching Theorem~\ref{thm:exact-case-tightness} within a factor of $\sqrt{(n+1)(t+1)}$.

\section{An approximate to exact result}
\label{s:repairUBLB}
In this section, we prove the following result, which develops bounds on the distance between a pair of symmetric $(k,\delta)$-wise indistinguishable distributions over $\{0,1\}^n$ and such a pair of $(k,0)$-wise indistinguishable distributions. We provide nearly matching upper and lower bounds. For symmetric distributions \(\mu,\nu\) on \(\{0,1\}^n\), we define the ``repair distance''
\[
    \operatorname{Rep}_{n,k}(\mu,\nu)
    :=
    \inf_{\mu',\nu'}
    \max\{\Delta(\mu,\mu'),\Delta(\nu,\nu')\},
\]
where the infimum is over symmetric distributions \(\mu',\nu'\) on
\(\{0,1\}^n\) which are \((k,0)\)-wise indistinguishable.

\begin{theorem}[Repair upper bound]
\label{thm:symmetric-repair-upper-precise}
Let \(0\le k<n\), and let \(\mu,\nu\) be symmetric
\((k,\delta)\)-wise indistinguishable distributions on \(\{0,1\}^n\).
Then
\[
    \operatorname{Rep}_{n,k}(\mu,\nu)
    \le \frac{U}{1+U},
    \qquad
    \text{where }
    U=\frac{\delta\sqrt{k+1}}{\lambda_{n,k,k}}.
\]
\end{theorem}

\begin{theorem}[Repair lower bound]
\label{thm:symmetric-repair-lower}
Let \(1\le k<n\) and for some sufficiently small constant $c$, set
\[
    \delta_0
    :=
    \min\left\{
        1,\,
        \lambda_{n,k,k}\cdot\frac{c}{k^{1/4}}
    \right\}.
\]
Then for every \(0<\delta\le \delta_0\), there exist symmetric
distributions \(\mu,\nu\) on \(\{0,1\}^n\) satisfying $\Delta_k(\mu,\nu)=\delta$
and hence \((k,\delta)\)-wise indistinguishability, such that
\[
    \operatorname{Rep}_{n,k}(\mu,\nu)
    \ge
    \frac{\delta}
         {2\lambda_{n,k,k}\sqrt{n+1}}.
\]
\end{theorem}
Theorems~\ref{thm:symmetric-repair-upper-precise}
and~\ref{thm:symmetric-repair-lower} show that the dependence on
\(\lambda_{n,k,k}^{-1}\) is necessary up to polynomial factors. For convenience, we note that applying Stirling's formula to Equations~\ref{eq:hahn-normalized-norm} and~\ref{eq:lambda-def} gives 
\[
    \lambda_{n,k,k}^{-1}
    =
    n^{O(1)}2^{\,k-n\cdot I(k/n)},
\]
from which the precise statement in Theorem~\ref{thm:symmetric-repair-upper} follows. 

\subsection{Notation and preliminaries}
For a signed measure \(\eta\) on a finite set \(\Omega\), we write
\[
    \eta(\Omega)
    :=
    \sum_{x\in\Omega}\eta(x)
\]
for its total mass, and
\[
    |\eta|(\Omega)
    :=
    \sum_{x\in\Omega}|\eta(x)|
\]
for its total variation mass. Thus, if \(\alpha,\beta\) are probability
measures on \(\Omega\), then
\[
    \Delta(\alpha,\beta)
    =
    \frac12|\alpha-\beta|(\Omega)
    =
    \frac12\sum_{x\in\Omega}|\alpha(x)-\beta(x)|.
\]
For a symmetric distribution \(\mu\) on \(\{0,1\}^n\), write
\[
    \overline{\mu}(j)
    :=
    \Pr_{X\sim\mu}[|X|=j],
    \qquad 0\le j\le n,
\]
as the distribution over Hamming weights induced by $\mu$.

We note that the marginalisation operator preserves total mass. Indeed, for any
signed measure \(\eta\) on \(\{0,\dots,n\}\),
\[
\begin{aligned}
    (M_{n\to t}\eta)(\{0,\dots,t\})
    &=
    \sum_{a=0}^t
    (M_{n\to t}\eta)(a)                                      \\
    &=
    \sum_{a=0}^t
    \sum_{j=0}^n
    K_{n,t}(j,a)\eta(j)                                      \\
    &=
    \sum_{j=0}^n
    \left(
        \sum_{a=0}^t K_{n,t}(j,a)
    \right)\eta(j)                                           \\
    &=
    \sum_{j=0}^n \eta(j)                                     \\
    &=
    \eta(\{0,\dots,n\}).
\end{aligned}
\]
Here the penultimate equality holds because, for each fixed
\(j\), the values $K_{n,t}(j,0),\dots,K_{n,t}(j,t)$ form a probability distribution.

\subsection{Proof of Theorem~\ref{thm:symmetric-repair-upper}}
\label{subsec:symmetric-approx-to-exact}    
\begin{proof}
Let $\sigma:=\overline{\mu}-\overline{\nu}$
be the signed difference of the Hamming-weight distributions and let $y:=M_{n\to k}\sigma$. We have that $\sigma$ is a difference of distributions and thus $\sigma(\{0,...,n\})=\sum_{j=0}^n \sigma(j)=0$. Since \(\mu,\nu\) are \((k,\delta)\)-wise indistinguishable and symmetric,
we have
\[
    |y|(\{0,...,k\})=\sum_{a=0}^k|y(a)|\le 2\delta.
\]

Regard \(y\) as a signed measure on \(\{0,\dots,k\}\), and let \(\widetilde{y}\) be the scaling 
\[
    \widetilde{y}(a):=(k+1)y(a),
    \qquad 0\le a\le k.
\]
Expand \(\widetilde{y}\) in the orthonormal Hahn basis on \(\{0,\dots,k\}\):
\[
    \widetilde{y}=\sum_{r=0}^k c_r\phi_r^{(k)}.
\]
The coefficients are
\[
    c_r
    =
    \langle \widetilde{y},\phi_r^{(k)}\rangle_k
    =
    \sum_{a=0}^k y(a)\phi_r^{(k)}(a).
\]
We now construct a signed correction \(\zeta\) on \(\{0,\dots,n\}\) such
that $M_{n\to k}\zeta=y$. To this end, define a function \(\widetilde{\zeta}:\{0,\dots,n\}\to\mathbb R\) by
\[
    \widetilde{\zeta}
    :=
    \sum_{r=0}^k
    \frac{c_r}{\lambda_{n,k,r}}\phi_r^{(n)}.
\]
Let \(\zeta\) be the signed measure on \(\{0,\dots,n\}\) defined by the scaling:
\[
    \zeta(j):=\frac{1}{n+1}\widetilde{\zeta}(j),
    \qquad 0\le j\le n.
\]
It follows immediately from the linearity of the marginalisation operator and Lemma~\ref{lem:hahn-measure-intertwining} that $M_{n\to k}\zeta=y$.

We next bound the total variation mass $|\zeta|(\{0,...,n\})$ of the correction \(\zeta\). By
orthonormality and Lemma~\ref{lem:lambda-monotone},
\[
    \|\widetilde{\zeta}\|_n^2
    =
    \sum_{r=0}^k
    \frac{c_r^2}{\lambda_{n,k,r}^2}\leq \frac{\|\widetilde{y}\|^2_k}{\lambda^2_{n,k,k}}
\]
Now,
\[
\begin{aligned}
    \|\widetilde{y}\|_k
    &=
    \left(
        \frac1{k+1}
        \sum_{a=0}^k ((k+1)y(a))^2
    \right)^{1/2}                                               \\
    &=
    \sqrt{k+1}
    \left(
        \sum_{a=0}^k y^2(a)
    \right)^{1/2}                                               \\
    &\le
    \sqrt{k+1}\sum_{a=0}^k|y(a)|                                          \\
    &\le
    2\delta\sqrt{k+1}.
\end{aligned}
\]
Therefore
\[
    \|\widetilde{\zeta}\|_n
    \le
    \frac{2\delta\sqrt{k+1}}{\lambda_{n,k,k}}.
\]
We then have that
\[
    |\zeta|(\{0,...,n\})= \frac{1}{n+1}\sum_{j=0}^n|\widetilde{\zeta}(j)|\leq\frac{1}{\sqrt{n+1}}\left(\sum_{j=0}^n\widetilde{\zeta}^2(j)\right)^{1/2}=\|\widetilde{\zeta}\|_n\leq\frac{2\delta\sqrt{k+1}}{\lambda_{n,k,k}}.
\]
Write the Jordan decomposition of $\zeta$ into positive and negative parts:
\[
    \zeta=\zeta^+-\zeta^-
\]
define $U:=\frac12|\zeta|(\{0,...,n\})$ and observe that 
\begin{equation}
\label{eq:aBound}
    \zeta^+(\{0,...,n\})=\zeta^-(\{0,...,n\})=U\leq \frac{\delta\sqrt{k+1}}{\lambda_{n,k,k}},
\end{equation}
where the first equality is equivalent to the claim that $\zeta(\{0,...,n\})=0$, which is justified via preservation of total mass under the marginalisation map and the fact that $\sigma$ is a difference of probability distributions:
\[
    \zeta(\{0,\dots,n\})
    =
    (M_{n\to k}\zeta)(\{0,\dots,k\})
    =
    y(\{0,\dots,k\})
    =
    (M_{n\to k}\sigma)(\{0,\dots,k\})
    =
    \sigma(\{0,\dots,n\})
    =
    0.
\]
Define new distributions as
\[
    \overline{\mu}' :=
    \frac{\overline{\mu}+\zeta^-}{1+U},
    \qquad
    \overline{\nu}' :=
    \frac{\overline{\nu}+\zeta^+}{1+U}.
\]
These are probability measures on \(\{0,\dots,n\}\), from which we can define \(\mu',\nu'\) as the corresponding symmetric distributions on \(\{0,1\}^n\).

Their signed Hamming-weight difference is
\[
    \overline{\mu}'-\overline{\nu}'
    =
    \frac{\overline{\mu}-\overline{\nu}-\zeta}{1+U}
    =
    \frac{\sigma-\zeta}{1+U}.
\]
Applying \(M_{n\to k}\), we obtain
\[
    M_{n\to k}(\overline{\mu}'-\overline{\nu}')
    =
    \frac{M_{n\to k}\sigma-M_{n\to k}\zeta}{1+U}
    =
    \frac{y-y}{1+U}
    =
    0.
\]
Thus \(\mu',\nu'\) are symmetric and have identical \(k\)-coordinate marginals. 

It remains to bound the distance from \(\mu,\nu\) to \(\mu',\nu'\).
For \(\mu\), we have

\[
\begin{aligned}
    \Delta(\overline{\mu}',\overline{\mu})
    &=
    \frac12
    \sum_{j=0}^n
    \left|
        \frac{\overline{\mu}(j)+\zeta^-(j)}{1+U}
        -
        \overline{\mu}(j)
    \right|                                                    \\[1ex]
    &=
    \frac12
    \sum_{j=0}^n
    \left|
        \frac{\overline{\mu}(j)+\zeta^-(j)
        -(1+U)\overline{\mu}(j)}{1+U}
    \right|                                                    \\[1ex]
    &=
    \frac{1}{2(1+U)}
    \sum_{j=0}^n
    \left|
        \zeta^-(j)-U\overline{\mu}(j)
    \right|                                                    \\[1ex]
    &\le
    \frac{1}{2(1+U)}
    \sum_{j=0}^n
    \left(
        |\zeta^-(j)|+U|\overline{\mu}(j)|
    \right)                                                    \\[1ex]
    &=
    \frac{1}{2(1+U)}
    \left(
        \zeta^-(\{0,\dots,n\})
        +
        U\,\overline{\mu}(\{0,\dots,n\})
    \right)                                                    \\[1ex]
    &=
    \frac{1}{2(1+U)}(U+U)                                      \\[1ex]
    &=
    \frac{U}{1+U}.
\end{aligned}
\]
Therefore $\Delta(\overline{\mu},\overline{\mu}')\le\frac{U}{1+U}$ and the same argument gives $\Delta(\overline{\nu},\overline{\nu}')\le\frac{U}{1+U}$.

By symmetry, $\Delta(\mu,\mu')=\Delta(\overline{\mu},\overline{\mu}')$ and $\Delta(\nu,\nu')=\Delta(\overline{\nu},\overline{\nu}')$. Hence $\Delta(\mu,\mu')$ and $\Delta(\nu,\nu')$ are each at most $\frac{U}{1+U}$, and the final inequality follows from Equation~\ref{eq:aBound}.
\end{proof}

\subsection{Proof of Theorem~\ref{thm:symmetric-repair-lower}}
For convenience, we define
\[
    L_i=|\theta_{i,k}|(\{0,\dots,i\}).
\]
We will need a lemma relating the ratio of $L_k$ to $L_n$, which is proven in the Appendix.
\begin{lemma}
\label{lem:l1-ratio-simple-lower}
Let \(1\le k\le n\). Then $\frac{L_k}{L_n}\ge\frac{2^k}{\sqrt{(k+1)\binom{2k}{k}}}=\Theta(k^{-1/4})$.
\end{lemma}

We now prove Theorem~\ref{thm:symmetric-repair-lower}. Our technique is to start with a single distribution and then perturb it in two different ways so that closeness of $k$-wise marginals is preserved (within distance $\delta$) and so that as much global discrepancy as possible is introduced. Our choice of perturbation is a properly scaled and parametrised signed Hahn measure. The base distribution that we perturb is the absolute value of the selected Hahn measure and is chosen to ensure that each distribution in our pair is nonnegative.

\begin{proof}(of Theorem~\ref{thm:symmetric-repair-lower})
Define the signed measure
\[
    \theta:=\lambda_{n,k,k}^{-1}\theta_{n,k},
\]
from which we have
\[
    |\theta|(\{0,\dots,n\})=\lambda_{n,k,k}^{-1}L_n.
\]
Since \(k\ge 1\), the function \(\phi_k^{(n)}\) is orthogonal to the
constant function. Hence $\theta(\{0,\dots,n\})=0$.

Let
\[
    \gamma:=\frac{2\delta}{L_k}.
\]
By the assumption \(\delta\le\delta_0\) and Lemma~\ref{lem:l1-ratio-simple-lower}, we have
\[
    \gamma
    \le
    \frac{2\lambda_{n,k,k}}{L_n}
    =
    \frac{2}{|\theta|(\{0,\dots,n\})}.
\]

Define a probability measure \(\rho\) on \(\{0,\dots,n\}\) by
\[
    \rho(j):=\frac{|\theta(j)|}{|\theta|(\{0,\dots,n\})}.
\]
Now define probability measures \(\overline{\mu},\overline{\nu}\) on
\(\{0,\dots,n\}\) by
\[
    \overline{\mu}:=\rho+\frac{\gamma}{2}\theta,
    \qquad
    \overline{\nu}:=\rho-\frac{\gamma}{2}\theta.
\]
These measures have total mass \(1\), because \(\rho\) has total mass
\(1\) and \(\theta\) has total mass \(0\). They are nonnegative because,
for every \(j\),
\[
    \rho(j)-\frac{\gamma}{2}|\theta(j)|
    =
    |\theta(j)|
    \left(
        \frac{1}{|\theta|(\{0,\dots,n\})}
        -
        \frac{\gamma}{2}
    \right)
    \ge 0.
\]
Let \(\mu,\nu\) be the symmetric distributions on
\(\{0,1\}^n\) corresponding to $\overline{\mu}$ and $\overline{\nu}$.

Their signed Hamming-weight difference is
\[
    \sigma:=\overline{\mu}-\overline{\nu}
    =
    \gamma\theta.
\]
It follows immediately from the linearity of the marginalisation operator and Lemma~\ref{lem:hahn-measure-intertwining} that $M_{n\to k}\theta=\theta_{k,k}$ and therefore $M_{n\to k}\sigma=\gamma\theta_{k,k}$.

Hence
\[
    \Delta_k(\mu,\nu)=
    \frac12 |M_{n\to k}\sigma|(\{0,\dots,k\})=
    \frac{\gamma}{2}|\theta_{k,k}|(\{0,\dots,k\}) =
    \frac{\gamma}{2}L_k =
    \delta.
\]
Thus \(\mu,\nu\) are \((k,\delta)\)-wise indistinguishable.

We now show that any repair to obtain exact \(k\)-wise indistinguishability requires global adjustment of at least the claimed amount. Let \(\mu',\nu'\) be any symmetric
\((k,0)\)-wise indistinguishable distributions on \(\{0,1\}^n\), and let
\(\overline{\mu}',\overline{\nu}'\) be their Hamming-weight laws. Define
\[
    \sigma':=\overline{\mu}'-\overline{\nu}',
    \qquad
    \eta:=\sigma-\sigma'.
\]
Since \(\mu',\nu'\) are exactly \(k\)-wise indistinguishable, $M_{n\to k}\sigma'=0$ and therefore $M_{n\to k}\eta=M_{n\to k}\sigma=\gamma\theta_{k,k}$.

We lower-bound the total variation mass of \(\eta\). Starting with the fact $1=\sum_{a=0}^k \phi_k^{(k)}(a)\theta_{k,k}(a)$ and using the
adjoint relationship~\eqref{eq:T-M-adjoint-measures}, gives
\[
\begin{aligned}
    \gamma
    &=
    \sum_{a=0}^k \phi_k^{(k)}(a)\,\gamma\theta_{k,k}(a)              \\
    &=
    \sum_{a=0}^k \phi_k^{(k)}(a)(M_{n\to k}\eta)(a)              \\
    &=
    \sum_{j=0}^n (T_{n,k}\phi_k^{(k)})(j)\eta(j)                 \\
    &=
    \lambda_{n,k,k}
    \sum_{j=0}^n \phi_k^{(n)}(j)\eta(j).
\end{aligned}
\]
Thus
\[
    \left|
        \sum_{j=0}^n \phi_k^{(n)}(j)\eta(j)
    \right|
    =
    \frac{\gamma}{\lambda_{n,k,k}}.
\]
Using $\|\phi_k^{(n)}\|_\infty\le \sqrt{n+1}$,
we obtain
\[
    \frac{\gamma}{\lambda_{n,k,k}}
    \le
    \|\phi_k^{(n)}\|_\infty\,|\eta|(\{0,\dots,n\})
    \le
    \sqrt{n+1}\,|\eta|(\{0,\dots,n\}).
\]
Therefore
\begin{equation}
\label{eq:eta-mass-lower-bound}
    |\eta|(\{0,\dots,n\})
    \ge
    \frac{\gamma}{\lambda_{n,k,k}\sqrt{n+1}}.
\end{equation}
Since all distributions are symmetric, statistical distance between the
full distributions equals statistical distance between their
Hamming-weight laws. Moreover,
\[
    \eta
    =
    (\overline{\mu}-\overline{\mu}')
    -
    (\overline{\nu}-\overline{\nu}').
\]
Hence
\[
\begin{aligned}
    |\eta|(\{0,\dots,n\})
    &\le
    |\overline{\mu}-\overline{\mu}'|(\{0,\dots,n\})
    +
    |\overline{\nu}-\overline{\nu}'|(\{0,\dots,n\})              \\
    &=
    2\Delta(\mu,\mu')
    +
    2\Delta(\nu,\nu')                                           \\
    &\le
    4\max\{\Delta(\mu,\mu'),\Delta(\nu,\nu')\}.
\end{aligned}
\]
Combining this with~\eqref{eq:eta-mass-lower-bound}, we get
\[
    \max\{\Delta(\mu,\mu'),\Delta(\nu,\nu')\}
    \ge
    \frac{\gamma}{4\lambda_{n,k,k}\sqrt{n+1}}
    =
    \frac{\delta}{2\lambda_{n,k,k}\sqrt{n+1}L_k}.
\]
Finally, since
\[
    L_k
    =
    |\theta_{k,k}|(\{0,\dots,k\})
    =
    \frac{1}{k+1}\sum_{a=0}^k |\phi_k^{(k)}(a)|
    \le
    \left(
        \frac{1}{k+1}\sum_{a=0}^k (\phi_k^{(k)}(a))^2
    \right)^{1/2}
    =
    1,
\]
we conclude that
\[
    d
    \ge
    \frac{\delta}{2\lambda_{n,k,k}\sqrt{n+1}}.
\]
This holds for every symmetric exactly \(k\)-wise indistinguishable pair
\(\mu',\nu'\). Taking the infimum over all such pairs gives
\[
    \operatorname{Rep}_{n,k}(\mu,\nu)
    \ge
    \frac{\delta}{2\lambda_{n,k,k}\sqrt{n+1}}.
\]
\end{proof}

\bibliographystyle{ACM-Reference-Format}
\bibliography{hahn_abi_final}

\appendix
\section{Technical lemmas}
We record here proofs of five technical lemmas.
\subsection{Proof of Lemma~\ref{lem:lambda-monotone}}
\begin{proof}
Because $\lambda_{n,t,r}$ is always positive, it suffices to show that $\lambda_{n,t,r+1}^{2}<\lambda_{n,t,r}^{2}$. Immediately from Equations~\ref{eq:hahn-normalized-norm} and~\ref{eq:lambda-def} we have that
\[
    \lambda_{n,t,r}^{2}
    =
    \prod_{j=0}^{r-1}
    \frac{(t-j)(n+j+2)}{(n-j)(t+j+2)},
\]
where the empty product is interpreted as \(1\) and as such we have that 
\[
    \frac{\lambda_{n,t,r+1}^{2}}{\lambda_{n,t,r}^{2}}=\frac{(t-r)(n+r+2)}{(n-r)(t+r+2)}<1,
\]
where the final inequality follows from
\[
\begin{aligned}
    &(n-r)(t+r+2) - (t-r)(n+r+2)  \\
    &\qquad =
    2(n-t)(r+1)>0.
\end{aligned}
\]
\end{proof}

\subsection{Proof of Lemma~\ref{lem:approx-controls-low-degree}}

\begin{proof}
Let
\[
    A_X:=\sum_{i=1}^{\ell} X_i,
    \qquad
    A_Y:=\sum_{i=1}^{\ell} Y_i.
\]
Since \(\mu\) and \(\nu\) are symmetric, the definition of \(T_{n,\ell}\) and the law of total
expectation yield
\[
    \mathbb{E}[h(A_X)]
    =
    \mathbb{E}\bigl[(T_{n,\ell}h)(\lvert X\rvert)\bigr]
    \qquad\text{and}\qquad
    \mathbb{E}[h(A_Y)]
    =
    \mathbb{E}\bigl[(T_{n,\ell}h)(\lvert Y\rvert)\bigr].
\]
Since \(\ell\le k\) and \(\mu,\nu\) are \((k,\delta)\)-wise
indistinguishable, by Equation~\ref{eq:stat_dis},
\[
\begin{aligned}
    \left|
        \mathbb E_{X\sim\mu}[(T_{n,\ell}h)(|X|)]
        -
        \mathbb E_{Y\sim\nu}[(T_{n,\ell}h)(|Y|)]
    \right|
    &=
    \left|
        \mathbb E[h(A_X)]
        -
        \mathbb E[h(A_Y)]
    \right| \\
    &\le
    2\delta\|h\|_\infty.
\end{aligned}
\]
\end{proof}

\subsection{Proof of Lemma~\ref{lem:nonasymptotic-damping}}
\begin{proof}
The case $r=0$ is immediate; we assume $r\geq 1$. Then, for \(0\le j< r\leq t<n\), we have
\[
    \frac{(t-j)(n+j+2)}{(n-j)(t+j+2)}
    =
    1-\varepsilon_j,
\]
where
\[
    \varepsilon_j
    :=
    \frac{2(n-t)(j+1)}
         {(n-j)(t+j+2)}.
\]
We also have directly from Equation~\ref{eq:lambda-def} that $\prod_{j=0}^{r-1}(1-\varepsilon_j)=\lambda^2_{n,t,r}$.
The value $\varepsilon_j$ is clearly non-negative, and is also at most 1 because $n-j\geq n-t$ and $t+j+2\geq2j+2$.
It is trivial that  
\[
    \varepsilon_j
    \ge
    \frac{2(n-t)(j+1)}
         {n(t+r+1)},
\]
and it follows that
\[
\begin{aligned}
    \frac{(t-j)(n+j+2)}{(n-j)(t+j+2)}
    &=
    1-\varepsilon_j                                                        \\
    &\le
    1-
    \frac{2(n-t)(j+1)}
         {n(t+r+1)}                                                    \\
    &\le
    \exp\!\left(
        -
        \frac{2(n-t)(j+1)}
             {n(t+r+1)}
    \right),
\end{aligned}
\]
where we used \(1-x\le e^{-x}\).

Multiplying this estimate over \(j=0,\dots,r-1\), we obtain
\[
\begin{aligned}
    \lambda_{n,t,r}^{2}
    &\le
    \exp\!\left(
        -
        \frac{2(n-t)}
             {n(t+r+1)}
        \sum_{j=0}^{r-1}(j+1)
    \right)                                                        \\
    &=
    \exp\!\left(
        -
        \frac{2(n-t)}
             {n(t+r+1)}
        \cdot
        \frac{r(r+1)}2
    \right)                                                        \\
    &=
    \exp\!\left(
        -
        \frac{(n-t)r(r+1)}
             {n(t+r+1)}
    \right).
\end{aligned}
\]
Taking square roots gives \eqref{eq:single-r-damping}.
\end{proof}

\subsection{Proof of Lemma~\ref{lem:lambda-inverse-bound}}
\begin{proof}
The case \(r=0\) is immediate, so assume \(r\ge 1\). Then, for $0<r\leq t<n$, we have from Equations~\ref{eq:hahn-normalized-norm} and~\ref{eq:lambda-def} we get
\[
    \lambda_{n,t,r}^{-2}=
    \prod_{j=0}^{r-1}
    \frac{(t+j+2)(n-j)}
         {(t-j)(n+j+2)}.
\]

For each \(0\le j\le r-1\), write the \(j\)-th factor as
\[
    \frac{(t+j+2)(n-j)}
         {(t-j)(n+j+2)}
    =
    1+
    \frac{
        (t+j+2)(n-j)-(t-j)(n+j+2)
    }
    {(t-j)(n+j+2)}
    =
    1+
    \frac{2(n-t)(j+1)}{(t-j)(n+j+2)}.
\]
Since \(j\le r-1\), we have $t-j\ge t-r+1$ and therefore
\[
    \frac{2(n-t)(j+1)}
         {(t-j)(n+j+2)}
    \le
    \frac{2(n-t)(j+1)}
         {n(t-r+1)}.
\]
Using \(\log(1+x)\le x\), we obtain
\[
\begin{aligned}
    \log \lambda_{n,t,r}^{-1}
    &=
    \frac12
    \sum_{j=0}^{r-1}
    \log\!\left(
        1+
        \frac{2(n-t)(j+1)}
             {(t-j)(n+j+2)}
    \right)                                                     \\
    &\le
    \frac12
    \sum_{j=0}^{r-1}
    \frac{2(n-t)(j+1)}
         {n(t-r+1)}                                         \\
    &=
    \frac{(n-t)}
         {n(t-r+1)}
    \sum_{j=0}^{r-1}(j+1)                                      \\
    &=
    \frac{(n-t)r(r+1)}
         {2n(t-r+1)}.
\end{aligned}
\]
Exponentiating gives
\[
    \lambda_{n,t,r}^{-1}
    \le
    \exp\!\left(
        \frac{(n-t)r(r+1)}
             {2n(t-r+1)}
    \right).
\]
\end{proof}

\subsection{Proof of Lemma~\ref{lem:l1-ratio-simple-lower}}
\begin{proof}
First,
\[
    L_n
    =
    \frac1{n+1}\sum_{j=0}^n |\phi_k^{(n)}(j)|
    \le
    \left(
        \frac1{n+1}\sum_{j=0}^n (\phi_k^{(n)}(j))^2
    \right)^{1/2}
    =
    1,
\]
where the last equality reflects the normalization of \(\phi_k^{(n)}\).

It remains to compute a lower bound for \(L_k\). By the standard finite-difference identity (cf. Section 5.3 of~\cite{GrahamKnuthPatashnik1989}), we have
\[
    \sum_{j=0}^k (-1)^j\binom{k}{j}p(j)=0,
\]
for any polynomial $p$ of degree at most $k-1$. By orthogonality of the Hahn polynomials $Q_r^{(k)}$ over $\{0,...,k\}$, it follows that $Q_k^{(k)}$ agrees with $v(j):=(-1)^j\binom{k}{j}$ up to scaling by a constant. By the fact that $v(0)=1=Q_k^{(k)}(0)$, we have for \(\{0,\dots,k\}\):
\[
    Q_k^{(k)}(j)=(-1)^j\binom{k}{j}.
\]
Therefore
\[
    \sum_{j=0}^k |Q_k^{(k)}(j)|=2^k.
\]
From Equations~\ref{eq:hahn-normalized-norm} and~\ref{eq:orthonormal-hahn} we have
\[
    \phi_k^{(k)}
    =
    \frac{Q_k^{(k)}}{\sqrt{H_{k,k}}},
    \qquad
    H_{k,k}
    =
    \frac{1}{k+1}\binom{2k}{k},
\]
and we get
\[
\begin{aligned}
    L_k
    &=
    \frac1{k+1}\sum_{j=0}^k |\phi_k^{(k)}(j)|  \\
    &=
    \frac1{k+1}\frac{2^k}{\sqrt{H_{k,k}}}      \\
    &=
    \frac{2^k}{\sqrt{(k+1)\binom{2k}{k}}}.
\end{aligned}
\]
Thus
\[
    \frac{L_k}{L_n}\ge L_k
    =
    \frac{2^k}{\sqrt{(k+1)\binom{2k}{k}}}=\Theta(k^{-1/4})
\]
\end{proof}

\end{document}